\documentclass{article}

\usepackage{amsmath,amssymb,graphicx,diagrams, amsthm}
\usepackage[colorlinks=true,citecolor=blue,urlcolor=blue,linkcolor=blue,bookmarksopen=false,unicode=true]{hyperref}
\usepackage[dvipsnames]{xcolor}
\newtheorem{theorem}{Theorem}[section]
\newtheorem{lemma}[theorem]{Lemma}
\newtheorem{definition}[theorem]{Definition}
\newtheorem{corollary}[theorem]{Corollary}

\newcommand{\argmin}{\mathrm{arg\, min}}

\def\rob{\mathrm{rob}}
\def\:{\colon}

\newcommand{\R}{\mathbb{R}}
\newcommand{\Q}{\mathbb{Q}}
\newcommand{\N}{\mathbb{N}}
\newcommand{\Z}{\mathbb{Z}}

\newcommand{\id}{\mathrm{id\,}}

\newcommand{\Long}[1]{}
\usepackage{ifthen}
\newcommand{\heading}[1]{\vspace{1ex}\par\noindent{\bf\boldmath #1}}

\begin{document}

\markboth{P. Franek, M. Kr\v c\'al}{Robust Satisfiability of Systems of Equations}

\title{Robust Satisfiability of Systems of Equations
\thanks{This is an extended and revised version of a paper that appeared in the
proceedings of the ACM-SIAM Symposium on Discrete Algorithms 2014.
This research was
supported by the Center of Excellence -- Inst.\ for Theor.\
Comput.\ Sci., Prague (project P202/12/G061 of GA~\v{C}R), by the Project LL1201 ERCCZ CORES
and by institutional support RVO:67985807.}}

\author{PETER FRANEK \\
{Institute of Computer Science, ASCR Prague} \\
MAREK KR\v C\'AL \\
{IST Austria, Klosterneuburg}}

\maketitle

\begin{abstract}
We study the problem of \emph{robust satisfiability} of systems of nonlinear equations, namely, whether for a given continuous function $f\:K\to\R^n$
on a~finite simplicial complex $K$ and $\alpha>0$, it holds that each function $g\:K\to\R^n$ such that $\|g-f\|_\infty \leq \alpha$,  has a root in $K$.
Via a reduction to the extension problem of maps into a sphere, we particularly show that this problem is 
decidable in polynomial time for every fixed $n$, assuming $\dim K \le 2n-3$. This is a substantial extension of previous computational applications of \emph{topological degree} and related concepts in numerical and interval analysis.

Via a reverse reduction we prove that the problem is undecidable when $\dim K\ge 2n-2$, where the threshold comes from the \emph{stable range} in homotopy theory.

For the lucidity of our exposition, we focus on the setting when $f$ is piecewise linear. Such functions can  approximate general continuous functions, and thus we get approximation schemes and undecidability of the robust satisfiability in other possible settings.
\end{abstract}

\section{Introduction}
\label{sec:intro}
In many engineering and scientific solutions, a highly desired
property is the resistance against noise or perturbations.
We can only name a fraction of the instances: stability in data
analysis~\cite{Carlsson:2009}, robust optimization~\cite{Ben:2009}, image processing~\cite{Goudail:2004},
or stability of numerical methods~\cite{Higham:2002}.
Some of the most crucial tools for robust design come from topology,
which can capture the stable properties of spaces and maps.
Famous concepts using topological methods in computer science
are fixed point theory~\cite{Cronin:1964}, fair division theory~\cite{Longueville:2012},
persistent homology~\cite{EdelsbrunnerHarer:PersistentTopologySurvey-2008}, or discrete Morse theory~\cite{Forman:2002}.

In this paper, we take the robustness perspective on solving systems of nonlinear equations, a fundamental problem 
in mathematics and computer science. The tools of algebraic topology  will enable us to quantify and compute the robustness of a solution---its resistance against perturbations of a given system of equations. The robustness of a root is a favourable property as the system  may come from imprecise measurements or from a model  with inherent uncertainty.\footnote{For example, settings of a problem may evolve over time and we would like to know if the problem has a solution also tomorrow.}
In case that the system is given by arithmetic expressions,
we can use less precise but fast floating point operations and
still identify the root, if it is a robust one.

Much research has been done in this direction, and as far as we know, the  \emph{topological degree} has always been the essential core of the approaches. The topological degree is indeed a powerful and practically usable tool for proving the existence of robust solutions of systems of the form $f=0$ for (nonlinear) continuous functions $f\:\R^n\to\R^n$. Developing techniques for such proofs is a major theme in the interval computation community
\cite{Neumaier:90,Kearfott:04,Frommer:05,Frommer:2007,Collins:08b}, although the emphasis is usually not put on the completeness of the
tests,\footnote{See \cite{Frommer:2007,Franek_Ratschan:2012} for some incompleteness results.} but on their usability within numerical solvers. Particularly, efficient formulas for topological degree has been devised in the case where the map $f$ is polynomial \cite{Eisenbud-degree,Szafraniec-degree}.  

\heading{Our contribution.} The main ingredient of our results is the replacement of the topological degree by the \emph{extendability} of maps into spheres. The extension problem is given by a continuous map $g$ defined on a subspace $A$ of a topological space  $X$ (in our case the map will always take values in a sphere $S^d$) and the question is whether $g$ can be continuously extended to whole of $X$.\footnote{The topological degree can be understood as a solution for a special case of the extension problem: namely, the degree of a map $\partial B^{d+1}\to S^d$ from the boundary of  a $(d+1)$-disk to a $d$-sphere is zero if and only if the map can be extended to whole of the disk.} Extendability provides a  solution that combines the following three properties:
\begin{figure}
\begin{center}
  \includegraphics[scale=.41,page=1]{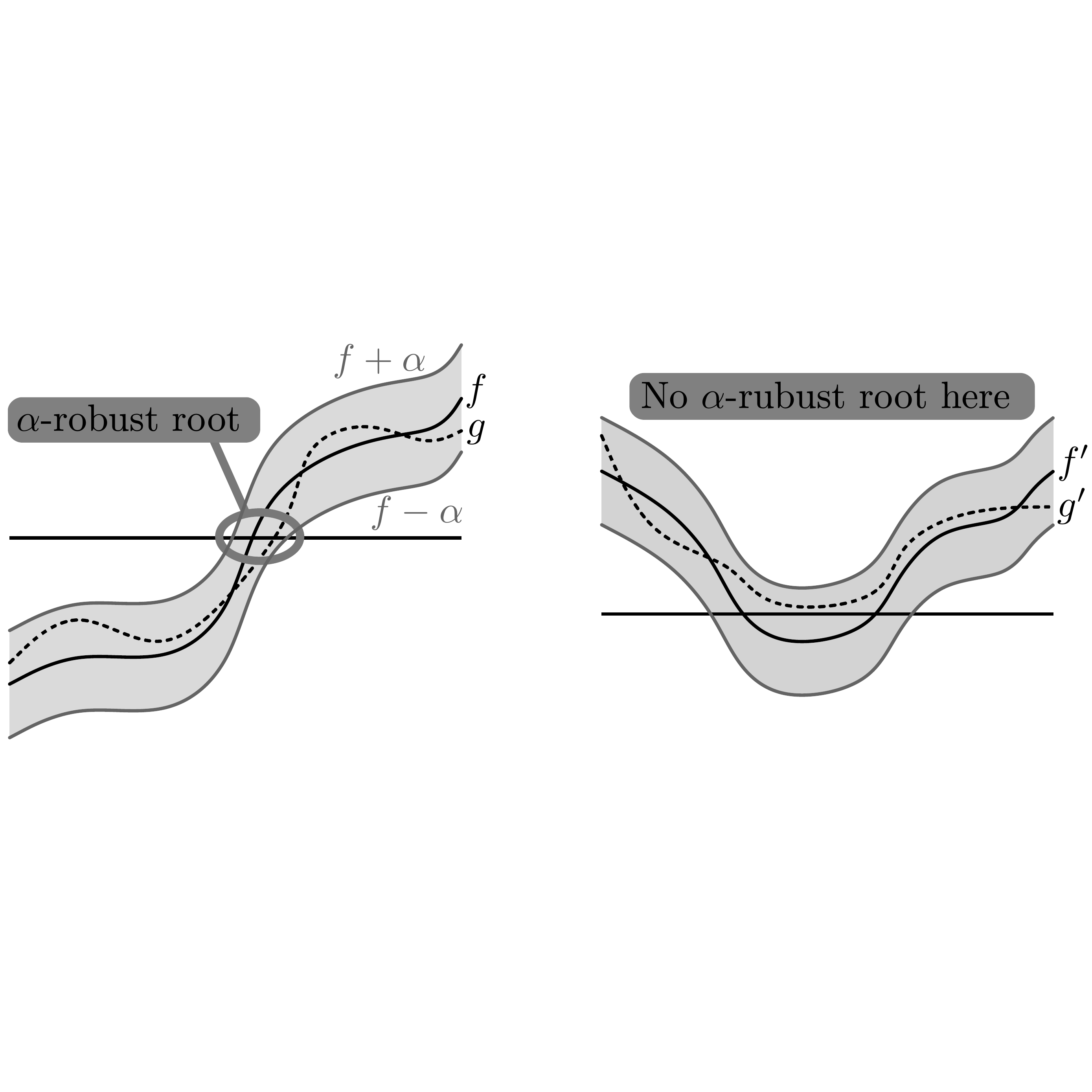}
  \includegraphics[scale=.41,page=1]{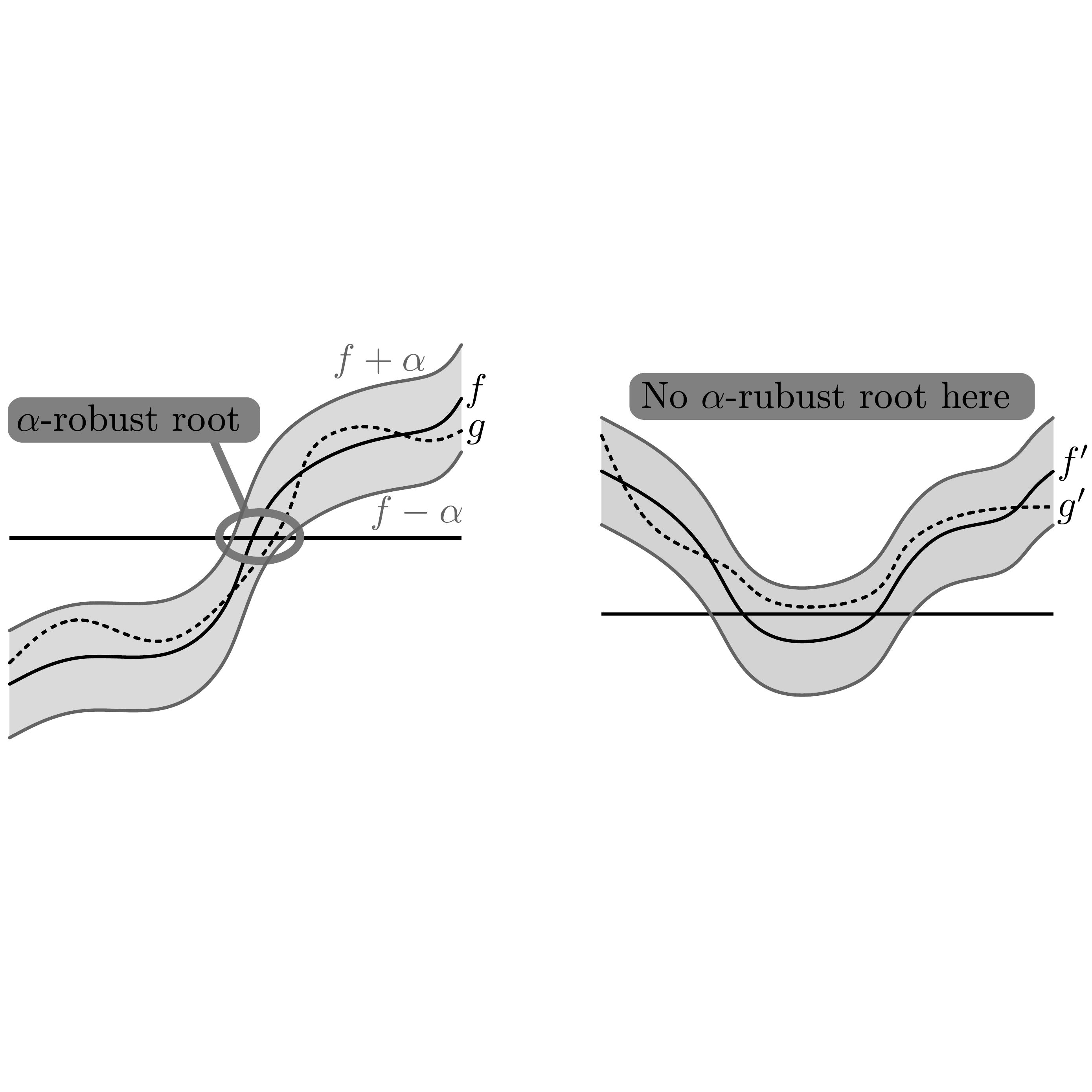}
  \caption{Illustration of one-dimensional functions with and without an $\alpha$-robust root.}
  \label{fig:1}
\end{center}
\end{figure}

\begin{enumerate}
\item
We \textbf{quantify the robustness} of the roots. Formally, for a given $\alpha>0$ we say that $f\:K\to\R^n$ on a compact domain $K\subseteq\R^m$ has an \emph{$\alpha$-robust root} whenever every \emph{$\alpha$-perturbation} $g$ of $f$---that is, a function $g\:K\to \R^n$ with $\|f-g\|:=\max_{x\in K} |f(x)-g(x)|\le \alpha$---has a root. 
\item
We show that the extendability is a \textbf{complete criterion} that verifies or disproves that a given function $f\:K\to \R^n$ has an $\alpha$-robust root for some given $\alpha>0$ (see Lemma~\ref{lemma:robust2ext}).
\item
Our solution includes \textbf{underdetermined equations},\footnote{Formally, our solution also includes overdetermined systems, but such systems never have a~robust solution.}
that is, the case with $\dim K>n$. On the positive side, we show an algorithm that works for  $\dim K\le 2n-3$ (here the running time is polynomial when $\dim K$ is fixed) or $n$ even (see Theorem~\ref{t:decidability}). On the negative side, we prove undecidability when $\dim K\ge2n -2$ and $n$ is odd (see Theorem~\ref{t:undecidability}).
\end{enumerate}
The core of the algorithm and the undecidability depends on recent  computational complexity studies of the extension problem~\cite{polypost,ext-hard,vokrinek:oddspheres}. We will state the relevant results in Section~\ref{top:prel}.

\heading{The computational setting.} 
In this paper, we assume the simple yet useful setting where the function $f\:K\to\R^n$ is piecewise linear (PL) on a~finite simplicial complex $K$.\footnote{A (geometric) simplicial complex $K$ is a collection of simplices in some $\R^d$ such that for every two simplices $\Delta,\Delta'\in K$ their intersection $\Delta\cap\Delta'$ is also a simplex of $K$ or an empty set. The underlying topological space $|K|:=\bigcup K$ will be also denoted by $K$ when no confusion can arise. } That means, $f$ is  linear on every simplex of $K$. Every such function $f$ is uniquely determined by its values on the vertices of $K$---as the linear interpolation on every simplex of $K$.
The usefulness of the PL setting can be seen in the facts that
\begin{itemize}
\item such~functions $f$ can be naturally obtained for instance  from measurements in every vertex of a grid or every vertex of a triangulation of a domain of interest, and,
\item arbitrary continuous function can be approximated arbitrarily precisely by a PL-function and thus our results in the PL setting yield analogous results in other possible settings, see Section~\ref{s:nonlinear}.
\end{itemize}
For easy computational representation of the input, we assume that our algorithms are given only PL functions with rational  values on the vertices.\footnote{We could use arbitrary subset of real numbers where the computations of absolute values and comparisons are possible in polynomial time.} 
We emphasize that the above mentioned set of $\alpha$-perturbations of a given function $f\:K\to \R^n$  contains continuous functions that are neither PL nor have rational values on the vertices. 

Every simplicial complex can be encoded as a hereditary set system $K$
on a set of vertices $V(K)$---such $K$ is called \emph{abstract simplicial complex}~\cite[p. 359]{Fulton:1995}.
Every abstract simplicial complex can be realized as a geometric simplicial complex uniquely 
up to a homoemorphism.\footnote{Every such a realization is given by a convenient embedding of the set of vertices in some 
$\R^d$. Then the collection of abstract simplices of $K$ determines the collection of the geometric simplices of the realization. 
The required condition is that  the intersection of arbitrary two geometric simplices of the realization also belongs to the realization or is empty.} 
We will use the notation $K$ for both the abstract simplicial complex and its geometric realization where no confusion is possible. 

Throughout this paper, $|\cdot |$ will denote a fixed norm on
$\R^n$. Its choice can be arbitrary for the undecidability results. For the algorithmic results we will need that 
\begin{eqnarray}\label{e:compare}
&&\begin{array}{l}
\text{for every rational point $y\in\Q^n$, the value $|y|$ can
be compared}\\
\text{to any given rational number $\alpha$, and,}\\
\end{array}\\
\label{e:min}&&\begin{array}{c}
\text{for every  PL map } f\:\Delta\to \R^n\text{ with rational values on the vertices} \\ 
\text{of a simplex }\Delta,  \text{ the point } \,\argmin_{x\in\Delta}|f(x)|\text{ can be computed.} \hfill
\end{array}
\end{eqnarray}
When the dimension of $\Delta$ is fixed, the computation can be done in polynomial time for $\ell_1$, $\ell_\infty$ (the simplex method for linear programming) and $\ell_2$ norms (Lagrange multipliers).
As already introduced above, for a function $f\:K\to\R^n$, the notation $\|f\|$ stands for the $L_\infty$ norm of $f$, that is, $\|f\|:=\max_{x\in K} |f(x)|$.

\heading{The algorithmic results.} The next definition introduces the main problem we study.
\begin{definition}
\label{def:rob-sat}
Assume that a norm on $\R^n$ is fixed. The Robust Satisfiability Problem (ROB-SAT) is the problem of deciding, for a given continuous function $f: |K|\to\R^n$
on a simplicial complex  $K$ and a number $\alpha > 0$, whether each $\alpha$-perturbation of $f$ has a root.

Let $\mathcal{E}$ be a set of functions with some agreed upon encoding.
We say that the ROB-SAT problem for $\mathcal{E}$ is \emph{decidable}, if there exists
an algorithm that correctly decides the above problem for each $f\in\mathcal{E}$ and $\alpha> 0$.
\end{definition}

\begin{theorem}\label{t:decidability}
 Let  $\mathcal{E}$ be the set of all PL functions $f\:K\to\R^n$ where $K$ is
a~finite simplicial complex, $n\in\N$ and either $\dim K\leq 2n-3$ or $n\leq 2$. Assume that a norm $|\cdot |$ on $\R^n$ is fixed and that we have an oracle for \eqref{e:compare} and \eqref{e:min}.\footnote{As we mentioned above, such an oracle can be implemented in polynomial time in the case of $\ell_1$, $\ell_2$ and $\ell_\infty$ norms.}
Then the ROB-SAT problem for $\mathcal{E}$ is decidable.
Moreover, for each fixed $n>0$, the running time is polynomial.

Additionally, ROB-SAT is decidable for all PL-functions $f\:K\to \R^n$ for $n$ even without any  restriction on the dimension of the simplicial complex $K$.
\end{theorem}
The decision procedure involves recent algorithmic results from computational homotopy theory, namely the algorithms for the \emph{extension problem}~\cite{ext-hard,vokrinek:oddspheres}.
The limit for the dimension of $K$ comes from so-called  \emph{stable range} in homotopy theory.

\heading{Locating a robust root.} In fact, our algorithm yields approximations of one or more  connected components within $K$ such that every $\alpha$-perturbation of $f$ has a root in each of the components.
Conversely, for every point $x$ of each of the components, there is an $\alpha$-perturbation with a root in $x$.
Thus the algorithm approximately localizes the roots within the precision that the  parameter $\alpha$ allows.

\heading{Computing the \emph{robustness of a root}.} Let  $f: K\to\R^n$ be a PL function.
It is easy to show that the set of all $\alpha\geq 0$, such that every $\alpha$-perturbation of $f$ has a root,
is either empty, or a closed and bounded set. Let us denote, if it is nonempty, 
the maximum of this set by $\rob(f)$---the \emph{robustness of the root of $f$}. We will prove that
for any PL function $f:K\to \R^n$ that has a root,  its robustness $\rob(f)$ is equal to $\min_{x\in\Delta} |f(x)|$ for some simplex $\Delta\in K$.
Consequently, by a simple binary search, we can compute the exact value $\rob(f)$.\footnote{Here we use the fact that if $|\cdot |$ is
the max-norm, then $|f|$ takes the minimum on each simplex in a computable rational point $x^*\in\Q^d$.
The same is true for the Euclidean norm $|\cdot |_2$,
but the minimum $\min_\Delta |f|_2$ is a square root of a rational number.}

It is worthwhile to identify other natural classes $\mathcal E$ of functions  on a compact domain for which the robustness of the roots can be computed exactly.
However, we omit this direction in this paper.
Arbitrary continuous functions on a compact domain can be approximated by piecewise linear ones up to arbitrary level of precision,
and thus the robustness of their roots can be approximated arbitrarily  precisely as well, see Section~\ref{s:nonlinear}.


\heading{Inequalities.}
We can generalize our results to systems of equations and \emph{inequalities}. Such system can be formally described
by $f=0\,\wedge\,g\leq 0$ for some $(f,g): K\to \R^n\times \R^k$. An \emph{$\alpha$-perturbation of this system} is
a system $\tilde{f}=0\,\wedge\,\tilde{g}\leq 0$ such that $\|f-\tilde{f}\|\leq\alpha$ and $\|g-\tilde{g}\|\leq\alpha$. Expectably, the \emph{robustness of a satisfiable system $f=0\,\wedge\,g\leq 0$} is defined to be the maximal $\alpha>0$ such that every $\alpha$-perturbation of the system has a solution.

If we assume that we use the max-norm in $\R^n$, then we will derive as a~corollary of Theorem~\ref{t:decidability}, that
the following problem is decidable: given PL functions $(f,g): K\to \R^n\times \R^k$ such that $\dim X\leq 2n-3$ or $n\leq 2$,
decide whether each $\alpha$-perturbation of $f=0\,\wedge\,g\leq 0$ is satisfiable or not. 
This will be discussed at the end of~Section~\ref{sec:decidability}.

\heading{Undecidability.}
The main question now is what happens in the remaining case when  $\dim K\ge 2n-2$ and $n$ is odd.
We claim that no  approximation of the robustness of roots is algorithmically possible here. The formal statement follows:
\begin{theorem}\label{t:undecidability}
Let $n>2$ be odd and $|\cdot|$ be an arbitrary norm in $\R^n$. Then there is no algorithm that, given a~finite
simplicial complex $K$ of dimension $2n-2$ and a PL function $f: K\to \R^n$, correctly chooses at least one of the following answers:
\begin{itemize}
\item $\rob(f)>0$ 
\item $\rob(f)< 1$
\end{itemize}
\end{theorem}
In particular, ROB-SAT for PL-functions is undecidable. 
Furthermore, it is easy to see that we have similar undecidability result for  any class of functions $\mathcal{E}$ such that, for any PL function $f$ and $\epsilon>0$,
we might algorithmically find $g\in\mathcal{E}$ such that $\|g-f\|<\epsilon$. For instance, the robustness of the roots cannot be approximated up to any constant factor for systems of polynomial equations, which is quite surprising in light of the fact that the first order theory $(\R, \leq)$ is decidable. While we can always decide the root existence problem for polynomials, we cannot, in general,
decide, whether the robustness of the root is greater than $\epsilon$ or less than $1$. 

The proof of Theorem~\ref{t:undecidability} is based on a recent undecidability result for the topological extension problem~\cite{ext-hard}.

%

If we consider systems of equations and \emph{inequalities}, an analogous undecidability result holds even for homotopically trivial domains.
Let $n>2$ be odd. Then there is no algorithm that, given $m,k\in\N$, a triangulation $T$ of $[-1,1]^m$ and a pair of PL functions $(f,g): T\to\R^n\times \R^k$, correctly chooses at least one of the following:
\begin{itemize}
\item the robustness of $f=0\,\wedge\,g\leq 0$ is greater than
$0$, or,
\item the robustness of $f=0\,\wedge\,g\leq 0$ is less than $1$.
\end{itemize}

An immediate consequence is that the ROB-SAT problem is undecidable for systems of polynomial equations and inequalities,
if all variables are from the interval $[-1,1]$ and no other constraint on the domain is given. 
The details are at the end of Section~\ref{sec:undec}.

\heading{Related work.} By the famous result of Tarski, the first order theory of real number is decidable.
On the countrary, in the unbounded case, the root existence problem is undecidable even for functions $f: \R\to\R$
that are expressed in closed forms containing polynomials and the $\sin$ function. A carefull examination of the
proof in~\cite{Wang:74} shows that also the ROB-SAT problem for such functions is undecidable.

As we mentioned, topological methods have been extensively used to identify roots of systems of equations,
particularly in numerical and interval analysis~\cite{Rall:80,Dian:03,Frommer:2004,Beelitz:05,Frommer:05,Smith:2012}.

Our work was  motivated by preceding papers \cite{Collins:08b,Franek_Ratschan:2012,degpaper}, where the topological degree was proved to be a complete  criterion for detecting \emph{infinitesimally robust} roots (i.e., the roots corresponding to the case $\rob(f)>0$) of nonlinear  systems of $n$ equations over $n$ variables. The approach extends to formulas obtained from systems of equations by adding the inequalities and the universal quantifiers.

The quantification of the robustness of the roots in our sense already appeared in the theory of \emph{well groups}~\cite{Edelsbrunner:2011}. 
Well groups measure in some sense the robustness of intersection of the image of a function $f: X\to Y$ with a closed subdomain $A\subseteq Y$.
As a special but important case with $X=\R^m, Y=\R^n$ and $A=\{0\}$, the zero-dimensional well group capture the robustness of a root of a function $f$. Particularly, Theorem~\ref{t:undecidability} implies incomputability of the well groups in the above special case with $m\geq2n-2$ and $n$ odd.
However, the  algorithms for computing well groups described in
\cite{ComputeWell,chazal} do not  reach beyond the case $m=n$, that is, the case with the same number of equations and variables.
See \cite{Skraba,Chazal:12} for some of the follow-up work.

We also mention the concept of \emph{Nielsen root number} which is a computable number defined for maps between manifolds of the same dimension.
It approximates a lower bound of the number of different solutions of $f(x)=c$~\cite{Brooks:73,Brown:01,Wong:2005}.
Like the topological degree, it is a homotopy invariant and does not change under small perturbations of $f$.
For pairs of maps $f_1, f_2: M^m\to N^n$ between compact manifolds of different dimensions $m,n$, Nielsen theory has been applied
to compute a lower bound for the number of connected components of the coincidence set $f_1'=f_2'$,
where $f_i'\simeq f_i$ are homotopic for $i=1,2$~\cite{Koschorke:2004,Koschorke:2007}.
Surprisingly, the author of \cite{Koschorke:2004} shows that under the condition $m\leq 2n-3$---the same as our stable dimension
range---a computable number $N(f_1, f_2)$ coincides with the minimal number of connected components of $\{f_1'=f_2'\}$ for
$f_1'$ resp. $f_2'$ homotopic to $f_1$ resp. $f_2$.

\heading{Open problems.} 
\begin{itemize}
\item
To what extent can be the present algorithms implemented and used on concrete instances? So far only the computation of the degree has been implemented and practically used for detecting robust roots\footnote{The topological degree gives a complete criterion in the case $\dim K=n$ but it can detect robust roots in other cases as well.} \cite{degpaper}. Can higher non-extendability tests via Steenrod squares and Adem's operation\footnote{The general algorithm for deciding extendability can be seen as a hierarchy of subsequent tests such that $\dim K - n$ of them is needed to get a complete answer. This algorithm has not been implemented yet, but in the case of a sphere, the second and the third tests can be obtained via Steenrod squares and Adem's operation \cite{MosherTangora:CohomologyOperations-1968}. Here the formulas are known and it might be relatively easy to get an efficient implementation.} be practically used?
\item
The undecidability result in Theorem~\ref{t:undecidability} applies to a general class of functions whose domains may have very complicated topology.
If we restrict the function space to functions defined on boxes (products of compact intervals), then the decidability of the ROB-SAT problem for such functions is open.
\item
A~natural generalization of the ROB-SAT problem is to consider first-order formulas, obtained from PL (in)equalities by adding conjunctions, negations and quantifiers.
For a rigorous statement see~\cite{Franek_Ratschan:2012}. Our current solution covers  the existentially quantified (conjunctions of) equations.
\item In the cases where we are able to verify the existence of a robust solution of $f=0$, a natural problem is to describe the zero set.
It is worth investigating the computable topological invariants of the solution set that are robust with respect to perturbations of $f$. In particular, can well groups be computed for $m\leq 2n-3$?
\end{itemize}
\section{Topological preliminaries}
\label{top:prel}
In this section we introduce some definitions from algebraic topology that we need throughout the proofs. The details can be found in standard
textbooks, such as~\cite{Hatcher,matveev2006lectures,spanier1994algebraic}.

We remind  that for a simplicial complex $K$, $|K|$ will refer to the underlying topological space, that is, the union of the simplices in $K$.
If no confusion can arise, we will denote by $K$ both the space and the simplicial complex itself.
All simplicial complexes are assumed to be finite, without explicitly saying it.
\heading{Star, link and subdivision.}
Let $A\subseteq K$ be simplicial complexes. We define the $star(A,K)$ to be the set of all faces of all simplices in $K$ that have
nontrivial intersection with $A$, and $link(A,K):=\{\sigma\in star(A,K)\,|\,\sigma\cap |A|=\emptyset\}$. Both $star(A,K)$ and $link(A,K)$ are simplicial complexes.
The difference $star^\circ(A,K):=|star(A,K)|\setminus |link(A,X)|$ is called \emph{open star}.
A simplicial complex $K'$ is called a \emph{subdivision} of $K$ whenever $|K'|=|K|$
and each $\Delta'\in K'$ is contained in some $\Delta\in K$.
If $a\in |K|$, than we may construct a subdivision of $K$ by replacing the unique $\Delta$ containing $a$ in its interior
by the set of simplices $\{a,v_1,\ldots, v_k\}$
for all $\{v_1,\ldots, v_k\}$ that span a face of $\Delta$, and correspondingly subdividing each simplex containing $\Delta$.  This process is called \emph{starring} $\Delta$ at $a$.
If we fix a point $a_\Delta$ in the interior of each $\Delta\in K$, we may construct a \emph{derived subdivision} $K'$ by starring each
$\Delta$ at $a_\Delta$, in an order of decreasing dimensions.
\heading{Homotopy extension property.}
We say that $f,g: X\to Y$ are \emph{homotopic}, if there exists a map $H: X\times[0,1]\to Y$ such that $H(\cdot,0)=f$ and $H(\cdot, 1)=g$.
The map $H$ is called a~\emph{homotopy}.
Let $A\subseteq X$ and $Y$ be topological spaces. A~map $F: X\to Y$ is called an~\emph{extension} of $f: A\to Y$, if the restriction $F|_A=f$.
The problem of deciding, whether there exists an extension $F$ of $f: A\to Y$, is called the \emph{topological extension problem}.
Let $H: A\times[0,1]\to Y$ be a homotopy. We say that the pair $(X,A)$ has the \emph{homotopy extension property} (HEP) with respect to $Y$,
if for any homotopy
$h: A\times [0,1]\to Y$ and an extension $F: X\to Y$ of $h(\cdot, 0)$, there exists a homotopy $H: X\times [0,1]\to Y$
that is an extension of~$h$.
The majority of common pairs $(X,A)$ possess this property with respect to any $Y$~\cite[p. 76]{Arkowitz:2011} and if $Y$ is finitely triangulable (e.g. the sphere),
then any closed subspace $A$ of a metric space $X$ has the HEP with respect to $Y$~\cite[p. 14]{Hu:59}.
If follows that the existence of an extension $F$ of $f: A\to Y$ depends only
on the \emph{homotopy class} of $f$.
\heading{Extendability of maps into a sphere.}
A map between simplicial complexes $X$ and $Y$ is called \emph{simplicial}, if the image of each simplex is a simplex and its restriction to each simplex is
a linear map.\footnote{This means, $f(\sum_j a_j v_j)=\sum_j a_j f(v_j)$ within a simplex with vertices $v_j$, $a_j\geq 0$ and $\sum_j a_j=1$.}
Every simplicial map $f\:X\to Y$ is uniquely defined by a mapping from the vertices of $X$ to the vertices of $Y$ (and thus can easily be described combinatorially).
Each continuous map $f: |K|\to |L|$ between simplicial complexes $K$ and $L$ can be approximated by a simplicial map
$f^\Delta: K'\to L$, where $K'$ is a subdivision of $K$. In particular, if $f(|star^\circ (v)|)\subseteq |star^\circ(f^\Delta(v))|$,
then $f$ and $f^\Delta$ are homotopic~\cite[p. 137]{Rotman:1988}.

A space $Y$  is called \emph{$(d-1)$-connected} iff every map $S^i\to Y$ for $i\le d-1$ is homotopic to a constant map. Every $d$-sphere $S^d$ is $(d-1)$-connected.
The core of the proofs of Theorems \ref{t:decidability} and \ref{t:undecidability} depends on the following recent results from computational homotopy theory.
\begin{theorem}[{\cite[Theorem 1.4]{polypost}}]
\label{top-existence}
Let $d\ge 2$ be fixed. Then there is a polynomial-time algorithm that, given finite simplicial complexes $X$, $Y$, a subcomplex
$A\subseteq X$, and a simplicial map $f\:A\to Y$,
where $\dim(X)\le 2d-1$ and $Y$ is $(d-1)$-connected,
decides whether $f$ admits an extension to a (not neccessarily simplicial) map $X\to Y$.
\end{theorem}
The bound $\dim(X)\leq 2d-1$ might seem too restrictive but it cannot be removed for even dimensional spheres, for example.
\begin{theorem}[{\cite[Theorem 1.1 (a)]{ext-hard}}]
\label{top-nonexistence}
Let $d  \geq 2$ be even and $S^{d}$ be a~simplicial complex representing the $d$-sphere.
Then the following problem is undecidable: given a simplicial pair
$(X,A)$ such that $\dim X=2d$ and a~simplicial map $f:A\to S^d$, decide whether there exists a continuous extension
$F: X \to S^d$ of $f$ or not.
\end{theorem}
We remark that in the case of $(d-1)$-connected spaces with $d$ odd, the undecidability holds for the wedge of two spheres $S^d\vee S^d$ \cite[Theorem 1.1 (b)]{ext-hard}. The interesting question on the complexity of the case of odd dimensional spheres, where homotopy groups $\pi_n(S^d)$, $n>d$ are all finite, was solved recently by Vok\v{r}\'{i}nek \cite{vokrinek:oddspheres}.

\begin{theorem}[{\cite[Theorem 1]{vokrinek:oddspheres}}]     \label{t:oddspheres}
There exists an algorithm that, given a pair of finite simplicial sets $(X, A)$, a finite $d$-connected simplicial set $Y$ , $d\geq 1$, with homotopy groups $\pi_n(Y)$ finite for all $2d \leq n \leq \dim X$ and a simplicial map $f\: A\to Y$ , decides the existence of a continuous extension $g: X \to Y$ of $f$.
\end{theorem}
We remark that there is no running time bound in Theorem \ref{t:oddspheres} even for fixed dimensions of $X$. We also believe that there is much more hope for practical implementations in the stable situation
(Theorem \ref{top-existence}) than in the unstable situation (Theorem \ref{t:oddspheres}).

We will need a slight extension of Theorem~\ref{top-existence} for maps to low-dimensional spheres $S^1$ and $S^0$.
The results of \cite{polypost} do not formally cover these cases mainly because they are too particular and not difficult at the same time.
We will need some basic terminology and facts from algebraic topology in the proof, such as homotopy and cohomology theory.
We point the interested reader to the textbooks  \cite{Hatcher, prasolov}.

\begin{lemma}
\label{low_dim}
There is a polynomial-time algorithm that, given a simplicial complex $X$ (of any dimension), $k\leq 1$
and a simplicial map $f:A\to S^k$ from a subcomplex $A$ of $X$, decides whether  there exists
a continuous extension (not necessarily simplicial) $F: X\to S^k$ of $f$.
\end{lemma}

\begin{proof}
First consider the case $k=0$.
We identify the connected components $X_1,X_2,\ldots,X_j$ of the complex $X$ and the subcomplexes $A_1,\ldots,A_j\subseteq A$ such that $A_i= X_i\cap A$ for $i=1,\ldots,j$.
It is straightforward that a simplicial map $f\:A\to S^0$ can be extended to $X$ iff $f$ is constant on every subcomplex $A_i$ for $i=1,\ldots,j$.

Now let $k=1$.
It is well known fact that the circle $S^1$ is the Eilenberg--MacLane space $K(\Z,1)$ and\footnote{The Eilenberg--MacLane space $K(\Z,1)$ is defined by $\pi_1(K(\Z,1))=\Z$ and $\pi_i(K(\Z,1))=0$ for every $i\neq 1$.
} that there is a natural bijection $[A,S^1]\to H^1(A;\Z)$.\footnote{For spaces $X$ and $Y$, we let $[X,Y]$ denote the set of homotopy classes of of maps $X\to Y$.} The bijection sends
the homotopy class of a continuous map $f\:A\to S^1$ to the cohomology class $f^*(\xi)$ where $f^*\:H^1(S^1;\Z)\to H^1(A;\Z)$ is the homomorphism induced by $f$ and $\xi\in H^1(S^1;\Z)$ is a certain distinguished cohomology class. 
In our case where $S^1=|\Sigma^1|,$ the cohomology class $\xi$ can be represented by a cocycle that assigns $1$ to the directed edge (\emph{ordered} simplex) $\overrightarrow{e_1e_2}$  of a~simplicial
complex $\Sigma^1$ and assigns $0$ to all the remaining directed edges. Thus $f^*(\xi)$ can be represented by a cocycle that assigns $1$ to every directed edge  $\overrightarrow{uv}$ such that $f(u)=e_1$ and $f(v)=e_2$ and $0$ to all the remaining ones. In the end, the question of extendability of a map $f\:A\to S^1$ reduces to extendability of the  cocycle on $A$ representing $f^*(\xi)$ to a cocycle on $X$.
This problem reduces to a system of linear Diophantine equations that can be solved in polynomial time \cite[Chapter~5]{Schrij86}.
\end{proof}

\section{Decidability}
\label{sec:decidability}
\heading{Realization of the sphere.}
Let us represent the sphere $S^{n-1}$ as the boundary of the cross  polytope, i.e., the convex hull of  $2n$ coordinate $\pm$ unit vectors $\pm e_1,\ldots, \pm e_n$.
This coincides with the set of unit vectors in $\R^n$ with respect to the $\ell_1$-norm $|x|_1:=\sum_i |x_i|$.
Let
$$r: \R^n\setminus\{0\}\to S^{n-1},\quad x\mapsto \frac{x}{|x|_1}.$$
This map is a \emph{homotopy equivalence} and the embedding $i: S^{n-1}\hookrightarrow \R^n\setminus\{0\}$ is its \emph{homotopy inverse}.
This implies that, for any simplicial complex pair $(X,A)$, there exists an extension $F: X\to \R^{n}\setminus\{0\}$
of $f: A\to \R^{n}\setminus\{0\}$ iff there exists an extension $\tilde{F}: X\to S^{n-1}$ of $r\circ f: A\to S^{n-1}$.

We choose the triangulation $\Sigma^{n-1}$ of $S^{n-1}$ consisting of all simplices spanned by subsets of the $2n$ vertices $\{\pm e_i\}$
that do not contain a pair of antipodal points $\{e_i, -e_i\}$.
For each index $i$ and sign $s$, $star(s e_i)$ is a triangulation of the hemisphere $\{x\in S^{n-1}\,|\,s\,x_i\geq 0\}$.

\heading{Robustness of the root.} Let $K$ be a compact space, $f: K\to \R^n$ be a continuous function and $g$ be an $\alpha$-perturbation of $f$ with no root.
Then $\min_K |g|>0$, so there exists an $\epsilon>0$ such that $\epsilon f+ (1-\epsilon) g$ has no root. However,
$\|\epsilon f + (1-\epsilon) g - f \|=(1-\epsilon) \|f-g\| \leq (1-\epsilon)\alpha$, so the set of all $\alpha'$ such that there exists
an $\alpha'$-perturbation of $f$ with no root, is an open set. It follows that the set of $\alpha$ such that each $\alpha$-perturbation of
$f$ has a root, is a closed set. It is clearly bounded, which justifies the following definition.
\begin{definition}
\label{def: robustness of f}
For a compact space $K$ and a continuous $f: K\to\R^n$ that has a root in $K$, we define
$$\rob(f):=\max \{\alpha\,|\,\text{each $\alpha$-perturbation of $f$ has a root}\}.$$
\end{definition}
For the purposes of this paper, we will need the following definition.
\begin{definition}
\label{def: critical}
Let $K$ be a simplicial complex and $f: K\to\R^n$ be a PL function. We will call $a\in\R$ a \emph{critical value of $f$}, if $a=\min_{x\in\Delta} |f(x)|$ for some $\Delta\in K$.
\end{definition}
We will show later that $\rob(f)$ is always a critical value of $f$.

\heading{Reduction of the ROB-SAT problem to the extension problem.}
Assume that $f: K\to\R^n$ is an arbitrary continuous function on a compact space $K$. 
There is the following general criterion for the existence of a robust root.

\begin{lemma}[Extendability criterion]\label{lemma:robust2ext}
A continuous function $f\:K\to\R^n$ has an $\alpha$-robust root if and only if its restriction to $|f|^{-1}\{\alpha\}$ cannot be extended to a map $|f|^{-1}[0,\alpha]\to\R^n\setminus\{0\}$.\end{lemma}
\begin{proof}
We will prove the equivalent statement that there exists an $\alpha$-per\-tur\-ba\-tion of $f$ with no root iff there exists an extension
$F: |f|^{-1}[0, \alpha]\to\R^n\setminus\{0\}$ of $f|_{|f|^{-1}\{\alpha\}}$.

Assume that $\tilde{f}: |f|^{-1}[0, \alpha] \to \R^n\setminus\{0\}$ is an $\alpha$-perturbation of $f$ with no root and
let $H(x, t):=t\tilde{f}(x) + (1-t) {f}(x)$ be a straight-line homotopy.
We will show that $H(x,t)\neq 0$ for $(x,t)\in |f|^{-1}\{\alpha\}\times [0,1]$. Using $|\tilde{f}(x)-{f}(x)|\leq\alpha$,
\begin{align*}
|H(x,t)|\,&= \,| t \tilde{f}(x) + (1-t) {f}(x)| \\
          &=|{f}(x)+t(\tilde{f}(x)-{f}(x))| \\
          &\geq |{f}(x)| - t |\tilde{f}(x)-{f}(x)| \\
          & \geq \alpha-t\alpha= (1-t)\alpha.
\end{align*}
This is positive for $t<1$. For $t=1$, $|H(x,1)|=|\tilde{f}(x)|>0$, because $\tilde{f}$ has no root.
So, $f|_{|f|^{-1}\{\alpha\}}$ and $\tilde{f}|_{|f|^{-1}\{\alpha\}}$ are homotopic maps from $|f|^{-1}\{\alpha\}$ to $\R^n\setminus\{0\}$.
The map $\tilde{f}: |f|^{-1}[0, \alpha]\to\R^n\setminus\{0\}$ is an extension of $\tilde{f}|_{|f|^{-1}\{\alpha\}}$
and due to the homotopy extension property of the pair $(|f|^{-1}[0,\alpha], |f|^{-1}\{\alpha\})$ wrt. the sphere,
there exists an extension $F: |f|^{-1}[0,\alpha]\to\R^n\setminus\{0\}$ of $f|_{|f|^{-1}\{\alpha\}}$.

Conversely, assume that there exists an extension $F: |f|^{-1}[0,\alpha]\to\R^n\setminus\{0\}$ of $f|_{|f|^{-1}\{\alpha\}}$.
Let $U$ be an open neighborhood of $|f|^{-1}\{\alpha\}$ in $|f|^{-1}[0, \alpha]$ such that for $x\in U$, $|F(x)-f(x)| < \alpha/2$.
Due to the compactness of $|f|^{-1}[0, \alpha]$, there exists $\epsilon\in (0,\alpha/2)$ such that $|f|^{-1}[\alpha-\epsilon, \alpha]\subseteq U$
(otherwise, there would exist a sequence $x_n\notin U$, $|f(x_n)|\to\alpha$ and a convergent subsequence $x_{j_n}\to x_0$,
whereas $x_0\in |f|^{-1} \{\alpha\}\subseteq U$, contradicting $x_{j_n}\notin U$). The map
$F|_{|f|^{-1}\{\alpha-\epsilon\}}$ is an $\alpha/2$-perturbation of $f|_{|f|^{-1}\{\alpha-\epsilon\}}$,
and an elementary calculation shows that they are homotopic as maps from $|f|^{-1}\{\alpha-\epsilon\}$ to $\R^n\setminus\{0\}$
by the straight-line homotopy. The homotopy extension property implies that $f_{|f|^{-1}\{\alpha-\epsilon\}}$
can be extended to a map $F_1: |f|^{-1}[0,\alpha]\to \R^n\setminus\{0\}$.
Without loss of generality, we may assume that $F_1(x)=f(x)$
for $x\in |f|^{-1}[\alpha-\epsilon, \alpha]$.

Let $g: |f|^{-1}[0,\alpha]\to (0,1]$ be a scalar function such that
$g(x)=1$ on $|f|^{-1}\{\alpha\}$ and $g(x)<\frac{\epsilon}{2 |F_1(x)|}$ on $|f|^{-1}[0, \alpha-\epsilon)$. Define the function
$F_2(x):=g(x) F_1(x)$ from $|f|^{-1}[0,\alpha]$ to $\R^n\setminus\{0\}$.
If $x\in |f|^{-1}[\alpha-\epsilon, \alpha]$, then
\begin{align*}
|F_2(x)-f(x)|&=|f(x) g(x) - f(x)|=|f(x)|\, |g(x)-1| \\
             &\leq |f(x)| < \alpha
\end{align*}
and if $x\in |f|^{-1}[0, \alpha-\epsilon)$, then
\begin{align*}
|F_2(x)-f(x)|&\leq |F_2(x)| + |f(x)|  \leq |F_1(x)| \frac{\epsilon}{2 |F_1(x)|} + \\
             & +\alpha-\epsilon = \alpha-\frac{\epsilon}{2} < \alpha
\end{align*}
which shows that $F_2$ is an $\alpha$-perturbation of $f$. By definition, $F_2(x)=f(x)$ for $x\in |f|^{-1}\{\alpha\}$, so
the function $F_3: K\to\R^n$ defined by
$$
F_3(x):=\begin{cases}
F_2(x)\quad\mathrm{for}\,\, x\in |f|^{-1} [0, \alpha] \\
f(x) \quad \mathrm{for}\,\, x\in |f|^{-1} (\alpha, \infty)
\end{cases}
$$
is a continuous $\alpha$-perturbation of $f$ with no root.
\end{proof}

The straightforward aim would be to apply the algorithm of Theorem~\ref{t:decidability} to the extension problem from the previous lemma. We would need to triangulate the spaces
$|f|^{-1}(\alpha)$ and $|f|^{-1}[0,\alpha]$---a feasible goal in the case $|\cdot|=|\cdot|_\infty$. However, we prefer a solution that covers larger class of norms, and in particular, the Euclidean norm. There it is less obvious how to work with the relevant spaces and maps. But we can use the fact that  extendability is invariant under homotopy equivalence. It will be enough to  apply the algorithm of Theorem~\ref{t:decidability} on ``combinatorial approximations'' of the spaces and maps from the previous lemma.

Let $f\:K\to\R^n$ be a PL function and assume, in addition, that
for each $\Delta\in K$, the function $|f|$ takes both the minimum and maximum on $\Delta$ in a vertex. This can be achieved for an arbitrary PL function by taking a derived subdivision of $K$. (The procedure will be detailed in the proof of Theorem~\ref{t:decidability}.)
We define an auxiliary PL function $\chi\:K\to\R$ by
$$\chi(v):=\begin{cases}0 & \text{when }|f(v)|<\alpha; \\
1/2 & \text{when } |f(v)|=\alpha; \\
1 & \text{when }|f(v)|>\alpha. \\
\end{cases} $$
Let us put $X:=\chi^{-1}[0,1/2]$ and $A:=\chi^{-1}(1/2)$.
\begin{lemma}[Combinatorial approximation]
\label{homeomorphism}
Let $X$ and $A$ be defined as above. Then there is a homeomorphism $h:X\to |f|^{-1}[0,\alpha]$ such that $h|_A$ is a homeomorphism $A\to|f|^{-1}\{\alpha\}$. Moreover, there is a homotopy $H:h\sim\id$ such that $f\circ H$ has no root on $A$.
\end{lemma}

\begin{proof}
The homeomorphism $h$ and the homotopy $H$ are defined simplexwise. Let $\Delta$ be a simplex of $K$.
Let the face of $\Delta$ spanned by vertices $v$ with
$\chi(v)=0$, resp. $\chi(v)=1$ , $\chi(v)=1/2$ be denoted by $\sigma$, resp. $\tau, \rho$.
First, if $\tau$ is empty, we define $H$  constant on every point of $\Delta=\Delta\cap X$. Second, if $\sigma$ is empty, then $\Delta\cap X=\rho=\Delta\cap|f|^{-1}\{\alpha\},$ as $\alpha$ is the minimum of $|f|$ on $\Delta$ and is attained in every point of $\rho$. In this case we define $H$ constant as well.
Note that if both $\rho$ and $\sigma$ were empty, the minimum of $|f|$ of $\Delta$ would be greater than $\alpha$ and thus $X\cap\Delta=\emptyset.$

Assume that $\sigma, \tau$ are nonempty and $x\in\Delta\cap X$. If $x\in\sigma$ or $x\in\rho$, then
we define the homotopy $H(x,t)=x$ consistently with the above description of $H$ on the faces $\tau$, $\rho$.
Assume that $x$ is neither in $\sigma$ nor in $\rho$.
If $\rho$ is empty, then there exist
unique numbers $a,b> 0$, $a+b=1$ such that $x=ay + bz$ for uniquely defined points $y\in\sigma$ and $z\in\tau$.
Similarly, if $\sigma, \tau, \rho$ are all nonempty, then there exist
unique coefficients $a,b,c> 0$, $a+b+c=1$ such that $x=ay+bz+cw$ for uniquely defined $y\in\sigma$, $z\in \tau$, $w\in\rho$.

We define a line segment $s$ by $s(t)=(1-c)(ty+(1-t)z)+cw$ for $t\in[0,1]. $ It is the segment  in $\Delta$ parallel to $\overline{yz}$ passing through $x$. Let $y':=s(0)$ denote the starting point and $z':=s(1)$ denote the endpoint of $s$. We observe that \begin{itemize}
\item
$s$ intersects $A$ in a unique point $x_A:=s(1/2)$;
\item $s$ intersects $|f|^{-1}\{\alpha\}$ in a unique point $x_\partial$. This holds because $|f|$ is convex on $\Delta$, $|f(y')|<\alpha$ and $|f(z')|>\alpha$.
\end{itemize}

\begin{figure}
\includegraphics[page=2]{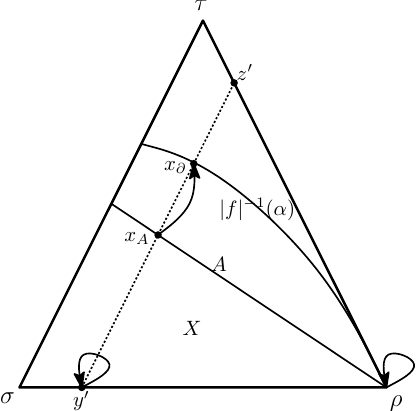}
\caption{Illustration of the Lemma~\ref{homeomorphism} in the case $|\cdot|=|\cdot|_2.$ The arrows represent the map $h=H(\cdot,1)$.
}
\label{fig: homeomorphism}
\end{figure}
A simple situation with all $\sigma,\tau$ and $\rho$ being singletons can be seen in Figure~\ref{fig: homeomorphism}.

We define
$$H(x,t):=x+ t(x_\partial-x_A)\frac{|y'-x|}{|y'-x_A|}.$$
This homotopy continuously stretches the segment $s$ in such a way that $y'$ is fixed and $x_A$ is sent to $x_\partial$ in $t=1$. We have that $H(\cdot, 0)$ is the identity
and the final map $h=H(\cdot, 1)$ is a~bijection---it bijectively stretches each segment $\overline{y'x_A}$ onto $\overline{y' x_\partial}$ (see Figure~\ref{fig: homeomorphism}).
For each $x\in \Delta\cap X$, $x\notin\sigma\cup\rho$, there exist unique $y', z', x_A$ and $x_\partial$ and they depend continuously on $x$.
Moreover, it is routine to check that the definition of $H$ on each $\Delta$ is compatible with its  definition on every face $\Delta'<\Delta$.
This altogether proves that both $H$ is continuous. So, the map $H(t,\cdot)$ is a continuous bijection of $X$ and its image
and the compactness of $X$ implies that $H(t,\cdot)^{-1}$ is continuous as well. In particular, $H(1,\cdot)=h$ is a continuous
bijection of $X$ and $|f|^{-1}[0,\alpha]$ and hence a homeomorphism.
Properties $h(|A|)\subseteq |f|^{-1}\{\alpha\}$ and $h^{-1}(|f|^{-1}\{\alpha\})\subseteq |A|$ follow from the construction.

Finally, once $\sigma$ is empty, then trivially $f\big(H(x,t)\big)\neq 0$. Otherwise for $x\in A$ the value $H(x,t)$
always belongs to the line segment $\overline{x_A x_\partial}$ that doesn't contain a root of $f$. 
\end{proof}
\begin{corollary}\label{col:homeomorphism}
There exists an extension $F: |f|^{-1}[0, \alpha]\to \R^n\setminus\{0\}$ of $f|_{|f|^{-1}\{\alpha\}}$
iff there exists  an extension $G: X\to\R^n\setminus\{0\}$ of $f|_A$.
\end{corollary}
\begin{proof}We proved that $h=H(1,\cdot)$ is a homeomorphism of the pairs $(X,A)$ and $(|f|^{-1}[0, \alpha],$  $|f|^{-1}\{\alpha\})$.
It follows that there exists an extension $F: |f|^{-1}[0, \alpha]\to\R^n\setminus\{0\}$ of $f|_{|f|^{-1}\{\alpha\}}$ iff there exists an extension
$G_1: X\to \R^n\setminus\{0\}$ of $(f\circ h)|_A$.
\begin{equation*}
\label{simpl_com_extension}
\begin{diagram}
X & \rTo_{h} & |f|^{-1}[0,\alpha] &  & \\
\uTo^{\rotatebox[origin=c]{90}{$\subseteq$}} & &  \uTo^{\rotatebox[origin=c]{90}{$\subseteq$}} & \rdDashto  \hskip 18pt F &\\
A & \rTo^{h} & |f|^{-1}\{\alpha\} & \rTo^{f} &  \R^n\setminus\{0\}
\end{diagram}
\end{equation*}
But $f\circ H|_{A\times [0,1]}$ is a homotopy between $f\circ h|_A: A\to \R^n\setminus\{0\}$ and $f|_A: A\to\R^n\setminus\{0\}$, so there exists an extension $G_1: X\to\R^n\setminus\{0\}$
of $(f\circ h)|_A$ iff there exists an extension $G: X\to\R^n\setminus\{0\}$ of $f|_A$. 
\end{proof}
\begin{corollary}
If $f: K\to\R^n$ is satisfiable, then the robustness $\rob(f)$
is a critical value of $f$.
\end{corollary}
\begin{proof}
Assume that $\alpha$ is not a critical value of~$f$ and each $\alpha$-perturbation of $f$ has a root. There are only finitely
many critical values, so there exists a $\beta>\alpha$ such that
$[\alpha,\beta]$ doesn't contain any critical value of $f$.
Assume that $|f|$ contains both the maximum and minimum on each simplex in a vertex, as
in the construction of Lemma~\ref{homeomorphism}. Then
for each vertex $v$ it holds that $|f(v)|<\alpha$ iff $|f(v)|<\beta$,
$|f(v)|>\alpha$ iff $|f(v)|>\beta$ and $\alpha\neq |f(v)|\neq\beta$,
so, the simplicial sets $A\subseteq X$ are identical for $\alpha$ and $\beta$.
It follows from Lemma~\ref{lemma:robust2ext} and~\ref{homeomorphism} $f: A\to\R^n\setminus\{0\}$ cannot
be extended to a map $X\to\R^n\setminus\{0\}$, and each $\beta$-perturbation of $f$ has a root as well. 
Thus $\alpha\neq \rob(f)$.
\end{proof}
%
\heading{The algorithm.}
It remains to show that the reduction from the $\alpha$-robust satisfiability of $f=0$ to the extension problem as
described in Theorem~\ref{top-existence} can be done in an algorithmic way. In order to do this, we will need to
convert the spaces $A\subseteq X$ constructed in Lemma~\ref{homeomorphism} into simplicial complexes and the PL function
$f|_A: A\to\R^n\setminus\{0\}$ into an equivalent simplicial map to the $(n-1)$-sphere.

In Lemma~\ref{homeomorphism} we assumed that $|f|$ contains the minimum on each simplex in some vertex.
This can be achieved by a starring each simplex $\Delta$ of $K$ at the point 
$\argmin_\Delta |f|$ whenever it belongs to the interior of $\Delta$. The point $\argmin_\Delta |f|$ can be computed by the assumption \eqref{e:min}. 

\begin{proof}[Proof of Theorem~\ref{t:decidability}]
First we compute $\argmin_\Delta |f|$ for each $\Delta\in K$ and construct a derived subdivision of $K$ by starring each
simplex $\Delta$ in $\argmin_\Delta |f|$.
Then $|f|$ contains the minimum on each simplex in a vertex.
Due to the linearity of $f$ on each simplex $\Delta$, $|f(ax+by)|=|a f(x) + b f(y)|\leq a |f(x)| + b |f(y)|$ for $x, y\in\Delta$,
$a,b\geq 0$, and $a+b=1$, so $|f|$ is a convex function and hence takes the maximum on each simplex in a vertex too.

In Lemma~\ref{lemma:robust2ext} we reduced the ROB-SAT problem for $f$ and $\alpha$ to the extendability of
$f|_{|f|^{-1}\{\alpha\}}:\, |f|^{-1}\{\alpha\}\to\R^n\setminus\{0\}$ 
to a function $|f|^{-1}[0,\alpha]\to\R^n\setminus\{0\}$.
By Corollary~\ref{col:homeomorphism}, this is equivalent to the extendability of $f|_A: A\to\R^n\setminus\{0\}$ to a function $X\to\R^n\setminus\{0\},$ where $X$ and $A$ are constructed as in Lemma~\ref{homeomorphism} with the use of the assumption \eqref{e:compare}. We will construct a subdivision $K'$ of $K$ containing a triangulation of $X$ and $A$.
To this end, $K'$ is constructed by starring each simplex $\Delta\in K$ that intersects $\chi^{-1}(1/2)$ in a point
$x_\Delta\in\Delta\cap\chi^{-1}(1/2)$. Now $\chi$ is a PL function on $K'$ that assigns values in $\{0,1/2,1\}$
to the vertices and, additionally,  no simplex has a pair of vertices evaluated to $0$ and to $1$.
Thus $X$ is spanned by vertices evaluated to $0$ or $1/2$, and $A$ is a simplicial subcomplex of $X$ spanned by
vertices evaluated to $1/2$.

The extendability of $f: A\to\R^n\setminus\{0\}$ to a function $X\to\R^n\setminus\{0\}$ is
further equivalent to the extendability of $r\circ f|_A: A\to S^{n-1}$ to a function $X\to S^{n-1}$, where $r(x):=x/|x|_1$ is the homotopy equivalence $\R^n\setminus\{0\}\to S^{n-1}$.

Next, we can algorithmically construct a subdivision $A'\subseteq X'$ of $A$ and $X$
such that for each $v\in A'$, there exists a pair $(i_v,s_v)\in \{1,\ldots, n\}\times \{+,-\}$ such that $f_{i_v}$
has constant sign $s_v$ on $star(v, A)$.
Let $\Sigma^{n-1}$ be a simplicial representation of the sphere, as defined at the beginning of Section~\ref{sec:decidability}.
Let $f^\Delta: A'\to \Sigma^{n-1}$ be a simplicial map that maps each vertex $v\in A'$ to $s_v \,e_{i_v}$.
This map is well defined, because if $a,b\in\Delta \in A'$, then $f^\Delta(b)\neq -f^\Delta(a)$, so $f^\Delta$ maps simplices to simplices.
It follows that $(r\circ f)(star(v, A))\subseteq star^\circ (f^\Delta(v), \Sigma^{n-1})$ for each
vertex $v\in A'$, so $f^\Delta \simeq r\circ f: A'\to S^{n-1}$ are homotopic by the simplicial approximation theorem.
The pair $(X',A')$ has the homotopy extension property, so there exists an extension $X\to S^{n-1}$ of $r\circ f|_A$
iff there exists an extension $X'\to S^{n-1}$ of the simplicial map $f^\Delta$.
In the case $\dim K<2n-2=2(n-1)$, it is decidable by Theorem~\ref{top-existence}, in the case $n\leq 2$, 
it is decidable by Lemma~\ref{low_dim} and for $n$ odd, it is decidable by Theorem~\ref{t:oddspheres}.

It remains to prove polynomiality for the cases $\dim K<2n-2$ or $n\leq 2$ for fixed $n$.
The input contains some encoding of the PL function $f$, which includes the information about $K$. 
In the algorithm, we first subdivided $K$ by starring it in $\argmin_\Delta{|f|}$ for each $\Delta\in K$ 
and then again by triangulating the sets $A$ and
$X$ defined in~\ref{homeomorphism}. 


Further, we construct a subdivision $X',A'$ of the pair $X,A$ such that $f$ can be represented by a simplicial map from $A'$. 
This can be achieved as follows. Let $i\in \{1,\ldots, n\}$. For each edge $ab$ of $A$ that intersects $f_i^{-1}\{0\}$ in its interior and is not contained in $f_i^{-1}\{0\}$, we choose a point $v_{ab}\in f_i^{-1}(0)\cap A$, and then star $X$ in each of this point (in some order). 
Consequently, the open star of each point in $A'$ is mapped by $f$ into a star of some vertex in $S^{n-1}$.
So for a fixed dimension $n$, we need only a fixed number of subdivisions of $(X,A)$ 
in order to construct the simplicial approximation $f^\Delta$ as in the proof of Theorem~\ref{t:decidability}. 

We reduced the ROB-SAT problem to the extension problem for $f^{\Delta}: A'\to S^{n-1}$
and $X'\supseteq A'$ and the number of simplices in $X'$ depends 
polynomially on the number of simplices in the original simplicial complex $K$.
Finally, we use the fact that the decision procedure for the extension problem for maps to spheres is polynomial in the number of simplices in $X'$, assuming that the dimension of $X'$ is fixed (Theorem~\ref{top-existence}).
\end{proof}

Note that if $f$ has a root, then the robustness $\rob(f)$ can be computed exactly by deciding, whether each $\alpha$-perturbation
of $f$ has a root for $\alpha$ ranging over the computable finite set of critical values. 
Moreover, if $n$ is fixed, then $\rob(f)$ can be computed in polynomial time.
\heading{Inequalities.} Here we show, how to reduce the robust satisfiability of systems of equations and inequalities to the ROB-SAT problem
for a system of equations. In this section, we assume that the norm $\|\cdot \|$ on $\R^n$-valued functions (implicitly used
in the definition of $\alpha$-perturbation) is derived from the 
max-norm in $\R^n$, i.e. $\|f\|:=\sup_{x\in K}\,\max_i |f_i(x)|$. Contrary to other parts of this paper, the proof of the 
following lemma uses the choice of the norm. 
\begin{lemma}
\label{lemma:ineq}
Let $(f,g):K\to \R^n\times\R^k$ and $\alpha>0$. Let $U:=\{x\in K\,|\,g(x)\leq -\alpha\}$~\footnote{The notation $g(x)\leq -\alpha$
means that $g_i(x)\leq-\alpha$ for each component $g_i$ of $g$.}.
Then each $\alpha$-perturbation of $f=0\,\wedge\,g\leq 0$ is satisfiable, 
iff each $\alpha$-perturbation $h: U\to \R^n$ of $f|_U$ has a root. 
\end{lemma}
\begin{proof}
The function $\tilde{g}(x):=g(x)+\alpha$ is an $\alpha$-perturbation of $g$. So, if each $\alpha$-perturbation of $f=0\,\wedge\,g\leq 0$
is satisfiable, then each $\alpha$-perturbation $\tilde{f}$ of $f$ has a root in $U=\{x\in K\,|\,\tilde{g}\leq 0\}$.

For the other implication, suppose that each $\alpha$-perturbation of $f$ has a root in $U$ and let
$\tilde{f}=0\,\wedge\,\tilde{g}\leq 0$ be an $\alpha$-perturbation of $f=0\,\wedge\,g\leq 0$.
Then $U\subseteq \{x\in K\,|\,\tilde{g}(x)\leq 0\}$, so $\tilde{f}=0$ has a solution on $\{x\in K\,|\,\tilde{g}(x)\leq 0\}$.
\end{proof}
If $f,g$ are PL function and $\dim K\leq 2n-3$ or $n\leq 2$, then 
we can construct a triangulation of the set $U=\{x\in K\,|\,g(x)\leq -\alpha\}$ and use Theorem~\ref{t:decidability}
to decide whether each $\alpha$-perturbation of $f$ has a root in $U$.

\section{Undecidability}
\label{sec:undec}
Here we  show that the robustness of roots of a PL function cannot be approximated in general.
We cannot algorithmically distinguish functions $f$ with $\rob(f)=0$
from functions $g$ with $\rob(g)\geq 1$.

We will continue to represent the sphere $S^{n-1}$ as a triangulation of $\{x\in\R^n\,|\,|x|_1=1\}$ as defined
at the beginning of Section~\ref{sec:decidability}.
All norms in $\R^{n}$ are equivalent, so for the norm $|\cdot |$  there exist numbers $\kappa_n>0$ and $\lambda_n>0$ such that
$|x|_1 \leq \kappa_n |x|$ and $|x| \leq \lambda_n |x|_1$ for all $x\in\R^n$. For the case of the max-norm,
$\kappa_n$ can be chosen to be $n$ and $\lambda_n$ can be chosen to be $1$. 
\begin{proof}[Proof of Theorem~\ref{t:undecidability}]
The proof proceeds by reduction from the extension problem as stated in Theorem~\ref{top-nonexistence}.
Let $X$ be a simplicial complex of dimension $2(n-1)$ and let $f: A\to \Sigma^{n-1}$ be a simplicial map from
some $A\subseteq X$ to the $(n-1)$-sphere. We assume that $n-1\geq 2$ is even, so by Theorem~\ref{top-nonexistence}, we cannot
algorithmically decide, given $X,A$ and $f$, whether there exists  an extension $X\to |\Sigma^{n-1}|$ of $f$ or not. 

We will assume that $A$ is \emph{full} in $X$, that is, for each simplex $\Delta$ it holds that $\Delta\in A$ iff
all vertices of $\Delta$ are in $A$.
(If $A$ is not full in $X$, then $X$ may be algorithmically subdivided into a complex $X'$ containing $A$, such that
$A$ is full in $X'$~\cite[Lemma 3.3]{rourke-sanderson:82}.)

Define a PL function $f': X\to\R^n$ given by its values on the vertices of $X$ as follows:
$$
f'(v):=\begin{cases}\kappa_n f(v); & \text{when }v\in A \\
0 & \text{when }v\notin A.
\end{cases}
$$
We will show the following statements:
\begin{enumerate}
\item[A.]
if there exists a~$1$-perturbation of $f'$ without a root, then $f$ can be extended to a function $X\to S^{n-1}$,
\item[B.]
if $f$ has an extension $F:X\to S^{n-1}$, then for every $\epsilon>0$ there exists an $\epsilon$-perturbation of $f'$ without a root.
\end{enumerate}
So, an algorithm defined in Theorem~\ref{t:undecidability} could decide the topological extension problem for
$f:A\to\Sigma^{n-1}$ and $X\supseteq A$, which is impossible by Theorem~\ref{top-nonexistence}.

It remains to show A. and B. We remind that for our realization of the sphere $S^{n-1}\subseteq\R^n$,
we have that $|f'(x)|=\kappa_n |f(x)|\geq |f(x)|_1=1$ for each $x\in |A|$.

\begin{enumerate}
\item[A.] Let $g'$ be a $1$-perturbation of $f'$ without a root.
It holds that $g'|_A$ is homotopic to $f'|_A=\kappa_n f$ as functions from $A$ to $\R^n\setminus\{0\}$ via the straight-line homotopy,
because
\begin{align*}
& |t g'(x) - (1-t) f'(x)|=|f'(x) + t (g'(x)-f'(x))| \\
& \geq |f'(x)| - t |g'(x)-f'(x)|\geq 1-t
\end{align*}
which is positive for $t<1$ and for $t=1$, $g'(x)\neq 0$ by assumption.
Since $g'$ is an extension of $g'|_A$, the map $f'$ can also be extended as a map into $\R^n\setminus\{0\}$ and $f$ can
be extended to a map from $|X|$ into $|\Sigma^{n-1}|=S^{n-1}$.
\item[B.]
We need the auxiliary PL function $\chi_A\:X\to \R$ that is $1$ on the vertices of $A$ and $0$ on the remaining vertices of $X$.
We will perturb $f'$ on a subdomain where it is ``tiny'',
namely, on $\chi_A^{-1}[0,\delta]$ for $\delta:=\epsilon /(2\,\kappa_n \lambda_n)$ and leave $f'$ unchanged on $\chi_A^{-1}[\delta,1]$ where no root occurs. Indeed, for every point $x$ with $\chi_A(x)>0$ either $x\in A$ (then $f'(x)\neq0$) or $x$ belongs to a simplex $\Delta\notin A$ and then $x$ lies in the interior of a segment $\overline{yz}$ where $$y\in\chi_A^{-1}(0) \text{ and } z\in |A|$$ (here we use that $A$ is full in $X$). That means, $x=cy+dz$ for $d>0$ and thus $f'(x)\neq 0$.

To define the perturbation we need an auxiliary continuous ``blow-up''
map $r\:\chi_A^{-1}[0,\delta]\to X$ defined on each $x=cy+dz$ with $c+d=1$, $d\le\delta$ and $y$ and $z$ as above (but here possibly $d=0$) by
$$r(x):= c'y+(d/\delta)z$$ where $c'+d/\delta=1$. The points $y$ and $z$ depend continuously on $x$, so $r$ is continuous as well.
\begin{figure}
\begin{center}
\includegraphics{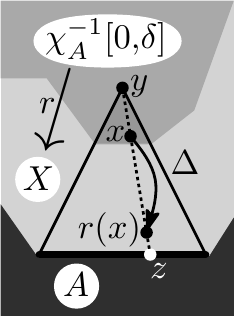}
\end{center}
\caption{Visualization of the map $r$ used in the proof of Theorem~\ref{t:undecidability}.}
\end{figure}

Let $F\:|X|\to |\Sigma^{n-1}|$ be the extension of $f$. We define an $\epsilon$-perturbation $g'$ of $f'$ by
$$g'(x):=\begin{cases}
\delta \, \kappa_n F(r(x)) & \text{when }x\in \chi_A^{-1}[0,\delta];\\
f'(x) & \text{otherwise.} \\
\end{cases}$$
To see that $g'$ is continuous, it is enough to verify that  $\delta \,\kappa_n\, F(r(x))=f'(x)$ for $x\in\chi_A^{-1}(\delta)$, that is, when $x=(1-\delta)y+\delta z$ for some $y$ and $z$ as above.
Indeed,
\begin{align*}
\delta\, \kappa_n F(r(x))&=\delta \, \kappa_n F(0y+1z)=\delta \,\kappa_n f(z),\quad\text{and}\\
f'(x)&=(1-\delta)\,f'(y)+\delta  f'(z)=0+\delta \,\kappa_n \,f(z).
\end{align*}
Finally, we check that $g'$ is an $\epsilon$-perturbation of $f'$ by the following estimation for every $x=cy+dz\in\chi_A^{-1}[0,\delta]$:
\begin{align*}
&|\delta\,\kappa_n F(r(x))-f'(x)| \leq \delta \, \kappa_n\|F\|+|f'(cy+dz)|\\
&\le \delta\,\kappa_n \lambda_n + d|f'(x)|\delta \,\kappa_n\,\lambda_n  + d\,\kappa_n |f(z)| \\
&\leq \delta\,\kappa_n\,\lambda_n + d \,\kappa_n \lambda_n \le 2\delta\kappa_n \lambda_n \leq \epsilon,
\end{align*}
because necessarily $d\leq\delta$. 
\end{enumerate}
\end{proof}

\heading{Inequalities.} An immediate consequence is that for systems of equations and inequalities, 
we get an undecidability result even for homotopically trivial domains such as products of compact intervals.
Formally, there is no algorithm that, given a triangulation $T$ of $[-1,1]^d$, PL functions $f,g: T\to \R^n$, correctly chooses at least one of the following options:
\begin{itemize}
\item the robustness of $f=0\,\wedge\,g\leq 0$ is greater than $0$, or,
\item the robustness of $f=0\,\wedge\,g\leq 0$ is less than $1$.
\end{itemize}

This is proved by a reduction from Theorem~\ref{t:undecidability} as follows. Let $K$ be a simplicial complex and $f: K\to\R^n$ a~PL function. 
We can algorithmically construct a PL embedding $K\hookrightarrow [-1,1]^d$ for some $d$
and a triangulation $T$ of $[-1,1]^d$ containing a subdivision $K'$ of $K$~\cite[p. 16]{rourke-sanderson:82}.
Furthermore, we assume that $K'$ is full in $T$ (otherwise we would subdivide $T$ once more). Define a scalar valued PL function $g: T\to\R$ to be $-1$ on the vertices of $K'$
and $1$ on the vertices of $T\setminus K'$. We immediately see that $\{x\,|\,g(x)\leq -1\}=|K'|$. 
Extend $f$ to a PL function $f_T: T\to\R^n$, by setting $f(v)=0$ for each vertex $T\setminus K'$.

If some $1$-perturbation of the system $f_T=0\,\wedge\,g\leq 0$ is not satisfiable, then it follows that some $1$-perturbation of $f$ has no root. Conversely, assume that, for some $\epsilon\in (0,1)$, 
some $\epsilon$-perturbation $\tilde{f}$ of $f$ has no root. Then we can extend $\tilde{f}$ to a function $\tilde{f}_T: T\to \R^n$
as follows:
\begin{itemize}
\item for $x\in T\setminus star (K',T)$, $\tilde{f}_T(x):=0$, and
\item for $x\in\Delta\in star(K',T)\setminus K'$, we have $x= t x_K + (1-t) x_T$ for uniquely determined 
$x_K\in \Delta\cap K'$, $x_T\in \Delta\cap (star(K',T)\setminus K')$ and $t>0$; then we define $\tilde{f}_T(x):=t \tilde{f}(x_K)$.
\end{itemize}
The resulting function $\tilde{f}_T$ is an $\epsilon$-perturbation of $\tilde{f}$, nonzero on the open star $star^\circ(K',T)$.
Clearly, $\{x\,|\,g(x)\leq 0\}\subseteq star^\circ(K',T)$ and thus $\tilde{f}_T=0\,\wedge\,g\leq 0$ is an
unsatisfiable $\epsilon$-perturbation of $f_T=0\,\wedge\,g\leq 0$.

We reduced the problem of deciding, whether $\rob(f)$
is greater that $0$ or less than $1$, to the problem of deciding, whether the robustness of $f_T=0\,\wedge\,g\leq 0$ is greater than $0$ or less than $1$. The former is undecidable by Theorem~\ref{t:undecidability},
so it follows that the latter is undecidable as well. 

\section{Nonlinear functions}
\label{s:nonlinear}
The decidability of the ROB-SAT problem for the class of PL functions defined in Theorem~\ref{t:decidability}
implies analogous results for larger function spaces that can be uniformly approximated by PL functions.
For example, any function $f: K\to\R^n$ containing expressions such as polynomials, division, $\sin$, $\exp$, $\log$ and $\pi$
can be uniformly approximated by a PL function up to any $\epsilon>0$ using \emph{interval arithmetic}~\cite{Moore:09}.
This implies the existence of an algorithm that takes such function $f: K\to\R^n$ and two constants $\alpha, \epsilon>0$,
and correctly chooses at least one of the following:
\begin{itemize}
\item Each $\alpha$-perturbation of $f$ has a root,
\item There exists an $(\alpha+\epsilon)$-perturbation of $f$ with no root,
\end{itemize}
assuming that $\dim K\leq 2n-3$. 
The decision procedure for this problem can construct an $\epsilon/2$-approximation PL approximation $f^{PL}$ of $f$, defined on
a triangulation $K_\Delta$ of $K$, which reduces the above problem to the ROB-SAT problem for $f^{PL}$ and $\alpha$.
Further, it can be shown that if $n$ is fixed and $f$ is smooth, then the size of $K_\Delta$ depends polynomially on $J/\epsilon$,
where $J$ is the upper bound on second partial derivatives $\partial f_i/\partial x_j$. This yields an estimate on the
algorithm complexity.
If $|K|$ is an $m$-box (product of $m$ compact intervals), than $f=0$ represents a system of $n$ equations in $m$ variables, with a given upper and lower bound for each variable.

Conversely, the undecidability result of Theorem~\ref{t:undecidability} generalizes to any class of  nonlinear  functions $\mathcal{E}$ such that any PL function can be algorithmically uniformly
approximated by some $g\in\mathcal{E}$. Given $\epsilon>0$, each $PL$ function $f$ can be algorithmically approximated by a component-wise polynomial function $p$
such that $||f-p||<\epsilon$, which immediately implies that the ROB-SAT problem is undecidable for general systems of polynomial equations defined on simplicial complexes,
once we exceed the dimensions of the stable range. The undecidability result can be slightly strengthened to functions defined on smooth manifolds, because each
pair of simplicial complexes $(X,A)$ can be algorithmically embedded to a pair of manifolds $(X^m, A^m)$ in some $\R^m$ with the same homotopical properties.
The interesting case of functions defined on $m$-boxes $[-1,1]^m$ was mentioned in the \emph{open problems} paragraph at the end of Section~\ref{sec:intro}.

\bibliographystyle{plain}
\bibliography{sratscha,Postnikov}
\end{document}